\documentclass[11pt, a4paper]{article}

\usepackage{amsmath}
\usepackage{amssymb}
\usepackage{amsfonts}
\usepackage{amsthm}
\usepackage{graphicx}
\usepackage{fullpage}
\usepackage{hyperref}
\usepackage{cite}

    \newtheorem{theorem}{Theorem}
    \newtheorem{proposition}[theorem]{Proposition}
    \newtheorem{lemma}[theorem]{Lemma}

    \theoremstyle{definition}
    
    \newtheorem{remark}{Remark}

\title{An Approximation Algorithm for Two-Edge-Connected Subgraph Problem via Triangle-free Two-Edge-Cover\thanks{%
This work was partially supported by the joint project of Kyoto University and Toyota Motor Corporation,
titled ``Advanced Mathematical Science for Mobility Society'', and
by JSPS KAKENHI Grant Numbers JP20K11692 and JP22H05001. 
}}

\author{Yusuke Kobayashi\thanks{Research Institute for Mathematical Sciences, Kyoto University.
E-mail: \{yusuke, tnoguchi\}@kurims.kyoto-u.ac.jp}
\and Takashi Noguchi\footnotemark[2]
}
\date{}

\begin{document}

\maketitle

\begin{abstract}
    The $2$-Edge-Connected Spanning Subgraph problem (2-ECSS) is one of the most fundamental and well-studied problems in the context of network design. 
    In the problem, we are given an undirected graph $G$, and the objective is to find 
    a $2$-edge-connected spanning subgraph $H$ of $G$ with the minimum number of edges.
    For this problem, a lot of approximation algorithms have been proposed in the literature.  
    In particular, very recently, Garg, Grandoni, and Ameli gave an approximation algorithm for 2-ECSS with factor $1.326$, 
    which was the best approximation ratio.     
    In this paper, we give a $(1.3+\varepsilon)$-approximation algorithm for 2-ECSS, where $\varepsilon$ is an arbitrary positive fixed constant, 
    which improves the previously known best approximation ratio. 
    In our algorithm, we compute a minimum triangle-free $2$-edge-cover in $G$ 
    with the aid of the algorithm for finding a maximum triangle-free $2$-matching given by Hartvigsen. 
    Then, with the obtained triangle-free $2$-edge-cover, we apply the arguments by Garg, Grandoni, and Ameli. 
\end{abstract}

\section{Introduction}
    In the field of survivable network design, a basic problem is to construct a network with minimum cost that satisfies a certain connectivity constraint. 
    A seminal result by Jain~\cite{Jain} provides a $2$-approximation algorithm for a wide class of survivable network design problems.
    For specific problems among them, a lot of better approximation algorithms have been investigated in the literature.      

    In this paper, we study the $2$-Edge-Connected Spanning Subgraph problem (2-ECSS), which is one of the most fundamental and well-studied problems in this context. 
    In 2-ECSS, we are given an undirected graph $G=(V, E)$, and the objective is to find 
    a $2$-edge-connected spanning subgraph $H$ of $G$ with the minimum number of edges.
    It was shown in~\cite{CL,CristG} that 2-ECSS does not admit a PTAS unless ${\rm P}={\rm NP}$.  
    Khuller and Vishkin~\cite{KV} gave a $3/2$-approximation algorithm for this problem, which was the starting point of the study of approximation algorithms for 2-ECSS. 
    Cheriyan, Seb\H{o}, and Szigeti~\cite{CSS} improved this ratio to $17/12$, and later 
    Hunkenschr\"{o}der, Vempala, and Vetta~\cite{VV,HVV} gave a $4/3$-approximation algorithm.
    By a completely different approach, Seb\H{o} and Vygen~\cite{SV} achieved the same approximation ratio. 
    Very recently, Garg, Grandoni, and Ameli~\cite{GGA} improved this ratio to $1.326$  
    by introducing powerful reduction steps and developing the techniques in~\cite{HVV}.
    
    The contribution of this paper is to present a $(1.3+\varepsilon)$-approximation algorithm for 2-ECSS for any $\varepsilon > 0$, 
    which improves the previously best approximation ratio. 
    \begin{theorem}
    \label{thm:main}
        For any constant $\varepsilon >0$, there is a deterministic polynomial-time $(1.3+\varepsilon)$-approximation algorithm for 2-ECSS.
    \end{theorem}

    Our algorithm and its analysis are heavily dependent on the well-developed arguments by Garg, Grandoni, and Ameli~\cite{GGA}. 
    In our algorithm, we first apply the reduction steps given in~\cite{GGA}. 
    Then, instead of a minimum $2$-edge-cover, we compute a minimum \emph{triangle-free} $2$-edge-cover in the graph, 
    which is the key ingredient in our algorithm. 
    We show that this can be done in polynomial time with the aid of the algorithm for finding a maximum triangle-free $2$-matching given by Hartvigsen~\cite{HartD} (see Theorem~\ref{thm:HartD}). 
    Finally, we convert the obtained triangle-free $2$-edge-cover into a spanning $2$-edge-connected subgraph 
    by using the arguments in~\cite{GGA}. 

    Our main technical contribution is to point out the utility of Hartvigsen's algorithm~\cite{HartD} in the arguments by Garg, Grandoni, and Ameli~\cite{GGA}. 
    It should be noted that Hartvigsen's algorithm has not received much attention in this context.

    \paragraph{Related Work}
    A natural extension of 2-ECSS is the $k$-Edge-Connected Spanning Subgraph problem ($k$-ECSS), 
    which is to find a $k$-edge-connected spanning subgraph of the input graph with the minimum number of edges.
    For $k$-ECSS, several approximation algorithms have been proposed, in which approximation factors depend on $k$~\cite{CT2000,GGTW2009,GG2012}.
    We can also consider the weigthed variant of 2-ECSS, in which the objective is to find 
    a $2$-edge-connected spanning subgraph with the minimum total weight in a given edge-weighted graph. 
    The result of Jain~\cite{Jain} leads to a $2$-approximation algorithm for the weighted 2-ECSS, and it is still the best known approximation ratio. 
    For the case when all the edge weights are $0$ or $1$, which is called the \emph{forest augmentation problem}, 
    Grandoni, Ameli, and Traub~\cite{GAT} recently gave a $1.9973$-approximation algorithm. 
    See references in~\cite{GGA,GAT} for more related work on survivable network design problems. 

    It is well-known that a $2$-matching of maximum size can be found in polynomial-time by using a matching algorithm; see e.g.,~\cite[Section 30]{lexbook}.
    As a variant of this problem, the problem of finding a maximum $2$-matching that contains no cycle of length at most $k$, 
    which is called the \emph{$C_{\le k}$-free $2$-matching problem}, has been actively studied. 
    Hartvigsen~\cite{HartD} gave a polynomial-time algorithm for the $C_{\le 3}$-free $2$-matching problem (also called the \emph{triangle-free $2$-matching problem}), and
    Papadimitriou showed the NP-hardness for $k \ge 5$ (see \cite{CP80}). 
    The polynomial solvability of the $C_{\le 4}$-free $2$-matching problem has been open for more than 40 years. 
    The edge weighted variant of the $C_{\le 3}$-free $2$-matching problem is also a big open problem in this area, 
    and some positive results are known for special cases~\cite{HL13,Kob10,PW21IPL,Kob22}. 
    See references in~\cite{Kob22} for more related work on the $C_{\le k}$-free $2$-matching problem.

\section{Preliminary}

Throughout the paper, we only consider simple undirected graphs, i.e., every graph has neither self-loops nor parallel edges.%
\footnote{It is shown in~\cite{HVV} that this assumption is not essential when we consider $2$-ECSS.} 
A graph $G=(V, E)$ is said to be \emph{$2$-edge-connected} if $G \setminus \{e\}$ is connected for any $e \in E$, 
and it is called \emph{$2$-vertex-connected} if $G \setminus \{v\}$ is connected for any $v \in V$ and $|V| \ge 3$. 
For a subgraph $H$ of $G$, its vertex set and edge set are denoted by $V(H)$ and $E(H)$, respectively. 
A subgraph $H$ of $G=(V,E)$ is \emph{spanning} if $V(H)=V(G)$. 
In the $2$-Edge-Connected Spanning Subgraph problem ($2$-ECSS), 
we are given a graph $G=(V, E)$ and the objective is to find a $2$-edge-connected spanning subgraph $H$ of $G$ with the minimum number of edges (if one exists). 

In this paper, a spanning subgraph $H$ is often identified with its edge set $E(H)$. 
Let $H$ be a spanning subgraph (or an edge set) of $G$. 
A connected component of $H$ which is 2-edge-connected is called a \emph{2EC component of $H$}.  
A 2EC component of $H$ is called an \emph{$i$-cycle 2EC component} if it is a cycle of length $i$. 
In particular, a $3$-cycle 2EC component is called a \emph{triangle 2EC component}.  
A maximal $2$-edge-connected subgraph $B$ of $H$ is called a \emph{block} of $H$
if $|V(B)| \ge 3$ and $B$ is not a 2EC component. 
An edge $e \in E(H)$ is called a \emph{bridge} of $H$ if $H \setminus \{e\}$ has more connected components than $H$. 
A block $B$ of $H$ is called a \emph{leaf block} if $H$ has exactly one bridge incident to $B$, 
and an \emph{inner block} otherwise. 

Let $G=(V, E)$ be a graph. For an edge set $F \subseteq E$ and a vertex $v \in V$, 
let $d_F(v)$ denote the number of edges in $F$ that are incident to $v$. 
An edge set $F \subseteq E$ is called a \emph{$2$-matching} if $d_F(v) \le 2$ for any $v \in V$, 
and it is called a \emph{$2$-edge-cover} if $d_F(v) \ge 2$ for any $v \in V$.%
\footnote{Such edge sets are sometimes called \emph{simple} $2$-matchings and \emph{simple} $2$-edge-covers in the literature.}

\section{Algorithm in Previous Work}
\label{sec:previous}

 Since our algorithm is based on the well-developed $1.326$-approximation algorithm given by Garg, Grandoni, and Ameli~\cite{GGA}, we describe some of their results in this section.  

\subsection{Reduction to Structured Graphs}

 In the algorithm by Garg, Grandoni, and Ameli~\cite{GGA}, they first reduce the problem to the case when 
 the input graph has some additional conditions, where such a graph is called a $(5/4,\varepsilon)$-structured graph. 
 In what follows in this paper, let $\varepsilon > 0$ be a sufficiently small positive fixed constant, which will appear in the approximation factor.  
 In particular, we suppose that $0\le \varepsilon \le 1/24$, which is used in the argument in~\cite{GGA}. 
 We say that a graph $G=(V,E)$ is \emph{$(5/4,\varepsilon)$-structured} if it is $2$-vertex-connected, it contains at least ${2}/{\varepsilon}$ vertices, and 
 it does not contain the following structures: 
        \begin{itemize}
            \item \textbf{($5/4$-contractible subgraph)} a $2$-edge-connected subgraph $C$ of $G$ such that every $2$-edge-connected spanning subgraph of $G$ contains at least $\frac{4}{5}|E(C)|$ edges with both endpoints in $V(C)$;
            \item \textbf{(irrelevant edge)} an edge $uv\in E$ such that $G \setminus \{u,v\}$ is not connected; 
            \item \textbf{(non-isolating $2$-vertex-cut)} a vertex set $\{u,v\}\subseteq V$ of $G$ such that $G\setminus\{u,v\}$ has at least three connected components or has exactly two connected components, both of which contains at least two vertices.
        \end{itemize}
The following lemma shows that it suffices to consider  $(5/4,\varepsilon)$-structured graphs when we design approximation algorithms. 

\begin{lemma}[\mbox{Garg, Grandoni, and Ameli~\cite[Lemma 2.2]{GGA}}]
\label{lem:structured}
For $\alpha \ge \frac{5}{4}$, 
if there exists a deterministic polynomial-time $\alpha$-approximation algorithm for 2-ECSS on $(5/4,\varepsilon)$-structured graphs, 
then there exists a deterministic polynomial-time $(\alpha+2\varepsilon)$-approximation algorithm for 2-ECSS.
\end{lemma}

\subsection{Semi-Canonical Two-Edge-Cover}

A $2$-edge-cover $H$ of $G$ (which is identified with a spanning subgraph) is called \emph{semi-canonical} if it satisfies the following conditions. 
    \begin{enumerate}
    \renewcommand{\theenumi}{(\arabic{enumi})}
    \renewcommand{\labelenumi}{\theenumi}
        \item \label{canonical:2EC} Each 2EC component of $H$ is a cycle or contains at least $7$ edges.   
        \item \label{canonical:block} Each leaf block contains at least $6$ edges and each inner block contains at least $4$ edges.
        \item \label{canonical:triangle} There is no pair of edge sets $F \subseteq H$ and $F' \subseteq E \setminus H$ such that $|F| = |F'| \le 3$, $(H \setminus F) \cup F'$ is a $2$-edge-cover with fewer connected components than $H$, and $F$ contains an edge in some triangle 2EC component of $H$.
        \item \label{canonical:4cycle} There is no pair of edge sets $F \subseteq H$ and $F' \subseteq E \setminus H$ such that $|F| = |F'| = 2$, $(H \setminus F) \cup F'$ is a $2$-edge-cover with fewer connected components than $H$, 
        both edges in $F'$ connect two $4$-cycle 2EC components, say $C_1$ and $C_2$, and $F$ is contained in $C_1 \cup C_2$. In other words, by removing $2$ edges and adding $2$ edges, 
        we cannot merge two $4$-cycle 2EC components into a cycle of length $8$.  
    \end{enumerate}

\begin{lemma}[\mbox{Garg, Grandoni, and Ameli~\cite[Lemma 2.6]{GGA}}]
\label{lem:fewtriangles}
  Suppose we are given a semi-canonical $2$-edge-cover $H$ of a $(5/4,\varepsilon)$-structured graph $G$ with $b|H|$ bridges and $t|H|$ edges belonging to triangle 2EC components of $H$. 
  Then, in polynomial time, we can compute a $2$-edge-connected spanning subgraph $S$ of size at most $(\frac{13}{10}+\frac{1}{30}t-\frac{1}{20}b)|H|$.     
\end{lemma}

\begin{remark}
 In the original statement of~\cite[Lemma 2.6]{GGA}, $H$ is assumed to satisfy a stronger condition than semi-canonical, called canonical. 
 A $2$-edge-cover $H$ is said to be \emph{canonical} if it satisfies \ref{canonical:2EC} and \ref{canonical:block}
 in the definition of semi-canonical $2$-edge-covers, and 
 also the following condition: 
 there is no pair of edge sets $F \subseteq H$ and $F' \subseteq E \setminus H$ such that $|F| = |F'| \le 3$ and $(H \setminus F) \cup F'$ is a $2$-edge-cover with fewer connected components than $H$.
 However, one can see that the condition ``canonical'' can be relaxed to ``semi-canonical'' by following the proof of~\cite[Lemma 2.6]{GGA}; see the proofs of Lemmas D.3, D.4, and D.11 in~\cite{GGA}. 
\end{remark}

\section{Algorithm via Triangle-Free Two-Edge-Cover}

The idea of our algorithm is quite simple: we construct a semi-canonical $2$-edge-cover $H$ with no triangle 2EC components and then apply Lemma~\ref{lem:fewtriangles}.  
We say that an edge set $F\subseteq E$ is \emph{triangle-free} if there is no triangle 2EC components of $F$. 
Note that a triangle-free edge set $F$ may contain a cycle of length three that is contained in a larger connected component. 
In order to construct a semi-canonical triangle-free $2$-edge-cover, 
we use a polynomial-time algorithm for finding a triangle-free $2$-matching given by Hartvigsen~\cite{HartD}. 

\begin{theorem}[\mbox{Hartvigsen~\cite[Theorem 3.2 and Proposition 3.4]{HartD}}]
\label{thm:HartD}
    For a graph $G$, we can find a triangle-free $2$-matching in $G$ with maximum cardinality in polynomial time.  
\end{theorem}

In Section~\ref{sec:trianglefree2ec}, we give an algorithm for finding a minimum triangle-free $2$-edge-cover with the aid of Theorem~\ref{thm:HartD}. 
Then, we transform it into a semi-canonical triangle-free $2$-edge-cover in Section~\ref{sec:canonical}. 
Using the obtained $2$-edge-cover, we give a proof of Theorem~\ref{thm:main} in Section~\ref{sec:proofmain}. 

\subsection{Minimum Triangle-Free Two-Edge-Cover}
\label{sec:trianglefree2ec}

As with the relationship between $2$-matchings and $2$-edge-covers (see e.g.~\cite[Section 30.14]{lexbook}), 
triangle-free $2$-matchings and triangle-free $2$-edge-covers are closely related to each other, 
which can be stated as the following two lemmas. 

\begin{lemma}
\label{lem:upper}
    Let $G=(V, E)$ be a connected graph such that the minimum degree is at least two and $|V| \ge 4$. 
    Given a triangle-free $2$-matching $M$ in $G$, in polynomial time, 
    we can compute a triangle-free $2$-edge-cover $C$ of $G$ with size at most $2|V|-|M|$. 
\end{lemma}

\begin{proof}
    Starting with $F=M$, we perform the following update repeatedly while $F$ is not a $2$-edge-cover:
    \begin{quote}
        Choose a vertex $v\in V$ with $d_F(v) < 2$ and an edge $vw\in E \setminus F$ incident to $v$. 
        \begin{enumerate}
            \renewcommand{\theenumi}{(\roman{enumi})}
            \renewcommand{\labelenumi}{\theenumi}
            \item \label{MtoC:1} If $F\cup \{vw\}$ is triangle-free, then add $vw$ to $F$. 
            \item \label{MtoC:2} Otherwise, $F\cup \{vw\}$ contains a triangle 2EC component with vertex set $\{u, v, w\}$ for some $u\in V$. 
            In this case, choose an edge $e$ connecting $\{u, v, w\}$ and $V \setminus \{u, v, w\}$, and add both $vw$ and $e$ to $F$. 
        \end{enumerate}
    \end{quote}
    If $F$ becomes a $2$-edge-cover, then the procedure terminates by returning $C = F$. 
    It is obvious that this procedure terminates in polynomial steps and returns a triangle-free $2$-edge-cover.  

    We now analyze the size of the output $C$. 
    For an edge set $F \subseteq E$, define $g(F) = \sum_{v \in V} \max \{2-d_F(v), 0\}$. 
    Then, in each iteration of the procedure, we observe the following: 
    in case \ref{MtoC:1}, one edge is added to $F$ and $g(F)$ decreases by at least one; 
    in case \ref{MtoC:2}, two edges are added to $F$ and $g(F)$ decreases by at least two, because $d_F(v) = d_F(w) = 1$ before the update.    
    With this observation, we see that $|C| - |M| \le g(M) - g(C) = \sum_{v\in V} (2-d_M(v))$, 
    where we note that $M$ is a $2$-matching and $C$ is a $2$-edge-cover. Therefore, it holds that 
    \begin{equation*}
        |C|\le |M|+\sum_{v\in V} (2-d_M(v))=|M|+(2|V|-2|M|)=2|V|-|M|, 
    \end{equation*}
    which completes the proof. 
\end{proof}

\begin{lemma}
\label{lem:lower}
    Given a triangle-free $2$-edge-cover $C$ in a graph $G = (V, E)$, in polynomial time, 
    we can compute a triangle-free $2$-matching $M$ of $G$ with size at least $2|V|-|C|$. 
\end{lemma}

\begin{proof}
   Starting with $F=C$, we perform the following update repeatedly while $F$ is not a $2$-matching:
    \begin{quote}
        Choose a vertex $v\in V$ with $d_F(v) > 2$ and an edge $vw\in F$ incident to $v$. 
        \begin{enumerate}
            \renewcommand{\theenumi}{(\roman{enumi})}
            \renewcommand{\labelenumi}{\theenumi}
            \item If $F\setminus \{vw\}$ is triangle-free, then remove $vw$ from $F$.
            \item If $F\setminus \{vw\}$ contains a triangle 2EC component whose vertex set is $\{v, v_1, v_2\}$ for some $v_1, v_2 \in V$, then remove $v v_1$ from $F$. 
            \item \label{FtoM:3} If neither of the above holds, then $F\setminus \{vw\}$ contains a triangle 2EC component whose vertex set is $\{w, w_1, w_2\}$ for some $w_1, w_2 \in V$. 
            In this case, remove $ww_1$ from $F$. 
        \end{enumerate}
    \end{quote}
    If $F$ becomes a $2$-matching, then the procedure terminates by returning $M = F$. 
    It is obvious that this procedure terminates in polynomial steps and returns a triangle-free $2$-matching.  

    We now analyze the size of the output $M$. 
    For an edge set $F \subseteq E$, define $g(F) = \sum_{v \in V} \max \{d_F(v)-2, 0\}$. 
    Then, in each iteration of the procedure, we observe that one edge is removed from $F$ and $g(F)$ decreases by at least one,  
    where we note that $d_F(w) = 3$ before the update in case~\ref{FtoM:3}. 
    With this observation, we see that $|C| - |M| \le g(C) - g(M) = \sum_{v\in V} (d_C(v) -2)$, 
    where we note that $C$ is a $2$-edge-cover and $M$ is a $2$-matching. Therefore, it holds that 
    \begin{equation*}
        |M|\ge |C| - \sum_{v\in V} (d_C(v)-2) = |C| - (2|C|-2|V|)=2|V|-|C|, 
    \end{equation*}    
    which completes the proof. 
\end{proof}

By using these lemmas and Theorem~\ref{thm:HartD}, we can compute a triangle-free $2$-edge-cover with minimum cardinality in polynomial time. 

\begin{proposition}
\label{prop:trifree2cover}
    For a graph $G=(V,E)$, we can compute a triangle-free $2$-edge-cover of $G$ with minimum cardinality in polynomial time (if one exists).
\end{proposition}

\begin{proof}
    It suffices to consider the case when $G$ is a connected graph such that the minimum degree is at least two and $|V| \ge 4$. 
    Let $M$ be a triangle-free $2$-matching in $G$ with maximum cardinality, which can be computed in polynomial time by Theorem~\ref{thm:HartD}.
    Then, by Lemma~\ref{lem:upper}, we can construct a triangle-free $2$-edge-cover $C$ of $G$ with size at most $2|V|-|M|$. 

    We now show that $G$ has no triangle-free $2$-edge-cover $C'$ with $|C'| < 2|V| - |M|$. 
    Assume to the contrary that there exists a triangle-free $2$-edge-cover $C'$ of size smaller than $2|V| - |M|$. 
    Then, by Lemma~\ref{lem:lower}, we can construct a triangle-free $2$-matching $M'$ of $G$ with size at least $2|V|-|C'|$.
    Since $|M'| \ge 2|V| - |C'| > 2|V| - (2|V| - |M|) = |M|$, this contradicts that $M$ is a triangle-free $2$-matching with maximum cardinality. 
    Therefore, $G$ has no triangle-free $2$-edge-cover of size smaller than $2|V| - |M|$, which implies that $C$ is a triangle-free $2$-edge-cover with minimum cardinality. 
\end{proof}

\subsection{Semi-Canonical Triangle-Free Two-Edge-Cover}
\label{sec:canonical}

We show the following lemma saying that 
a triangle-free $2$-edge-cover can be transformed into a semi-canonical triangle-free $2$-edge-cover without increasing the size. 
Although the proof is almost the same as that of~\cite[Lemma 2.4]{GGA}, we describe it for completeness. 

\begin{lemma}
\label{lem:convert}
 Given a triangle-free $2$-edge-cover $H$ of a $(5/4, \varepsilon)$-structured graph $G = (V, E)$, 
 in polynomial time, we can compute a triangle-free $2$-edge-cover $H'$ of no larger size which is semi-canonical.   
\end{lemma}

\begin{proof}
    Recall that an edge set is identified with the corresponding spanning subgraph of $G$. 
    Starting with $H' = H$, while $H'$ is not semi-canonical we apply one of the following operations in this order of priority. 
    We note that $H'$ is always triangle-free during the procedure, and hence it always satisfies condition~\ref{canonical:triangle} in the definition of semi-canonical $2$-edge-cover. 
    \begin{enumerate}
        \item[(a)]  If there exists an edge $e \in H'$ such that $H' \setminus \{e\}$ is a triangle-free $2$-edge-cover, then remove $e$ from $H'$. 

        \item[(b)] If $H'$ does not satisfy condition~\ref{canonical:4cycle}, then
            we merge two $4$-cycle 2EC components into a cycle of length $8$ 
            by removing $2$ edges and adding $2$ edges. 
            Note that the obtained edge set is a triangle-free $2$-edge-cover that has fewer connected components. 

        \item[(c)]
            Suppose that condition~\ref{canonical:2EC} does not hold, i.e., there exists a 2EC component $C$ of $H'$ with fewer than $7$ edges that is not a cycle.
            Since $C$ is $2$-edge-connected and not a cycle, we obtain $|E(C)| \ge |V (C)| + 1$. 
            If $|V(C)|=4$, then $C$ contains at least $5$ edges and contains a cycle of length $4$, 
            which contradicts that (a) is not applied. 
            Therefore, $|V(C)| = 5$ and $|E(C)| = 6$. 
            Since operation (a) is not applied, $C$ is either a bowtie (i.e., two triangles that share a commmon vertex) or 
            a $K_{2,3}$; see figures in the proof of~\cite[Lemma 2.4]{GGA}. 
            \begin{enumerate}
                \item[(c1)] Suppose that $C$ is a bowtie that has two triangles $\{v_1, v_2, u\}$ and $\{v_3, v_4, u\}$. 
                If $G$ contains an edge between $\{v_1, v_2\}$ and $\{v_3, v_4\}$, then we can replace $C$ with a cycle of length $5$, which decreases the size of $H'$. 
                Otherwise, by the $2$-vertex-connectivity of $G$, there exists an edge $zw \in E \setminus H'$ such that $z \in V \setminus V(C)$ and $w \in \{v_1, v_2, v_3, v_4\}$. 
                In this case, we replace $H'$ with $(H' \setminus \{uw\}) \cup \{zw\}$. 
                Then, the obtained edge set is a triangle-free $2$-edge-cover with the same size, which has fewer connected components. 
                
                \item[(c2)] Suppose that $C$ is a $K_{2, 3}$ with two sides $\{v_1, v_2\}$ and $\{w_1, w_2, w_3\}$. 
                If every $w_i$ has degree exactly $2$, then every feasible $2$-edge-connected spanning subgraph contains all the edges of $C$, and hence $C$ is a $\frac{5}{4}$-contractible subgraph, 
                which contradicts the  assumption that $G$ is $(5/4, \varepsilon)$-structured. 
                If $G$ contains an edge $w_i w_j$ for distinct $i, j \in \{1, 2, 3\}$, then we can replace $C$ with a cycle of length $5$, which decreases the size of $H'$. 
                Otherwise, since some $w_i$ has degree at least $3$, 
                there exists an edge $w_i u \in E \setminus H'$ such that $i \in \{1, 2, 3\}$ and $u \in V \setminus V(C)$.  
                In this case, we replace $H'$ with $(H' \setminus \{v_1 w_i\}) \cup \{w_i u\}$. 
                Then, the obtained edge set is a triangle-free $2$-edge-cover with the same size, which has fewer connected components.                 
            \end{enumerate}

        \item[(d)]
            Suppose that the first half of condition~\ref{canonical:block} does not hold, i.e., there exists a leaf block $B$ that has at most $5$ edges. 
            Let $v_1$ be the only vertex in $B$ such that all the edges connecting $V(B)$ and $V \setminus V(B)$ are incident to $v_1$. 
            Since operation (a) is not applied, we see that $B$ is a cycle of length at most $5$. 
            Let $v_1, \dots , v_\ell$ be the vertices of $B$ that appear along the cycle in this order.
            We consider the following cases separately; see figures in the proof of~\cite[Lemma 2.4]{GGA}.
            \begin{enumerate}
                \item[(d1)] Suppose that there exists an edge $zw \in E \setminus H'$ such that $z \in V \setminus V(B)$ and $w \in \{v_2, v_\ell\}$.  
                    In this case, we replace $H'$ with $(H' \setminus \{v_1 w\}) \cup \{z w\}$.

                \item[(d2)] Suppose that $v_2$ and $v_\ell$ are adjacent only to vertices in $V(B)$ in $G$, which implies that $\ell \in \{4, 5\}$. 
                    If $v_2 v_\ell \not\in E$, then every feasible 2EC spanning subgraph contains four edges (incident to $v_2$ and $v_\ell$) with both endpoints in $V(B)$, 
                    and hence $B$ is a $\frac{5}{4}$-contractible subgraph, which contradicts the assumption that $G$ is $(5/4, \varepsilon)$-structured. 
                    Thus, $v_2 v_\ell \in E$. 
                    Since there exists an edge connecting $V \setminus V(B)$ and $V(B) \setminus \{v_1\}$ by the $2$-vertex-connectivity of $G$,
                    without loss of generality, we may assume that $G$ has an edge $v_3 z$ with $z \in V \setminus V(B)$. 
                    In this case, we replace $H'$ with $(H' \setminus \{v_1 v_\ell, v_2 v_3\}) \cup \{v_3 z, v_2 v_\ell\}$.
            \end{enumerate}
            In both cases, the obtained edge set is a triangle-free $2$-edge-cover with the same size. 
            Furthermore, we see that either (i) the obtained edge set has fewer connected components or
            (ii) it has the same number of connected components and fewer bridges. 
                 
        \item[(e)]
            Suppose that the latter half of condition~\ref{canonical:block} does not hold, i.e., there exists an inner block $B$ that has at most $3$ edges. 
            Then, $B$ is a triangle. Let $\{v_1, v_2,  v_3\}$ be the vertex set of $B$. 
            If there are at least two bridge edges incident to distinct vertices in $V(B)$, say $wv_1$ and $z v_2$, then 
            edge $v_1 v_2$ has to be removed by operation (a), which is a contradiction. 
            Therefore, all the bridge edges in $H'$ incident to $B$ are incident to the same vertex $v\in V(B)$. 
            In this case, we apply the same operation as (d). 
    \end{enumerate}

    We can easily see that each operation above can be done in polynomial time. 
    We also see that each operation decreases the lexicographical ordering of $(|H'|, {\rm cc}(H'), {\rm br}(H'))$, 
    where ${\rm cc}(H')$ is the number of connected components in $H'$ and 
    ${\rm br}(H')$ is the number of bridges in $H'$. 
    This shows that the procedure terminates in polynomial steps.  
    After the procedure, $H'$ is a semi-canonical triangle-free $2$-edge-cover with $|H'| \le |H|$, which completes the proof. 
\end{proof}

\subsection{Proof of Theorem~\ref{thm:main}}
\label{sec:proofmain}

    By Lemma~\ref{lem:structured}, in order to prove Theorem~\ref{thm:main}, 
    it suffices to give a $\frac{13}{10}$-approximation algorithm for 2-ECSS in $(5/4, \varepsilon)$-structured graphs 
    for a sufficiently small fixed $\varepsilon > 0$. 
    Let $G=(V, E)$ be a $(5/4, \varepsilon)$-structured graph. 
    By Proposition~\ref{prop:trifree2cover}, 
    we can compute a minimum-size triangle-free $2$-edge-cover $H$ of $G$ in polynomial-time. 
    Note that the optimal value ${\sf OPT}$ of 2-ECSS in $G$ is at least $|H|$, because 
    every feasible solution for 2-ECSS is a triangle-free $2$-edge-cover. 
    By Lemma~\ref{lem:convert}, $H$ can be transformed into a semi-canonical triangle-free $2$-edge-cover $H'$ with $|H'| \le |H|$. 
    Since $H'$ is triangle-free, by applying Lemma~\ref{lem:fewtriangles} with $H'$,   
    we obtain a $2$-edge-connected spanning subgraph $S$ of size at most $(\frac{13}{10}-\frac{1}{20}b)|H'|$, where $H'$ has $b|H'|$ bridges. 
    Therefore, we obtain 
    \[
    |S| \le \left( \frac{13}{10}-\frac{1}{20}b \right) |H'| \le \frac{13}{10} |H| \le \frac{13}{10} {\sf OPT}, 
    \]
    which shows that $S$ is a $\frac{13}{10}$-approximate solution for 2-ECSS in $G$. 
    This completes the proof of Theorem~\ref{thm:main}. \qed

\section{Concluding Remarks}
    In this paper, we have presented a $(1.3+\varepsilon)$-approximation algorithm for 2-ECSS,  
    which achieves the currently best approximation ratio. 
    We give a remark that our algorithm is complicated and far from practical, because 
    we utilize Hartvigsen's algorithm~\cite{HartD}, which is quite complicated.  
    Therefore, it will be interesting to design a simple and easy-to-understand approximation algorithm with (almost) the same approximation ratio as ours. 
    Another possible direction of future research is to further improve the approximation ratio by improving Lemma~\ref{lem:fewtriangles}.

\bibliographystyle{plain}
\bibliography{ref}

\end{document}